%% file: fullversion.tex
\documentclass{article}

\usepackage{times}
\usepackage{graphicx} 
\usepackage{subfigure} 

\usepackage{natbib}

\usepackage{algorithm}
\usepackage{algorithmic}

\usepackage{booktabs}
\usepackage{array,multirow}
\usepackage{arydshln}
\usepackage{hyperref}

\usepackage[accepted]{icml2017} 

\usepackage{amsmath}
\usepackage{amsthm}
\usepackage{color}
\usepackage{mathtools,amsfonts,enumerate}
\usepackage{thmtools, thm-restate}
\usepackage[titletoc,title]{appendix}

\usepackage{enumitem}

\newtheorem{theorem}{Theorem}

\newenvironment{sproof}{%
  \proof}{\endproof}
  
\newcommand{\vol}{{\operatorname{vol}}}
\newcommand{\ex}{\operatorname{ex}}
\newcommand\myeq{\mathrel{\overset{\makebox[0pt]{\mbox{\normalfont\tiny\sffamily def}}}{=}}}
\newcommand{\polylog}{\operatorname{polylog}}
\newcommand{\tab}{.\hskip.1in}

\usepackage{tikz}
\usetikzlibrary[graphs]
\usetikzlibrary{matrix}
\usetikzlibrary{shapes,arrows,chains,positioning}
\usepgflibrary{arrows} 
\usepgflibrary[arrows] 
\usepgflibrary[plotmarks] 
\usetikzlibrary{arrows} 

\usetikzlibrary{decorations,decorations.pathmorphing,decorations.markings}
\usetikzlibrary{fit,calc,positioning,decorations.pathreplacing,matrix}
\begin{document} 

\twocolumn[

\icmltitle{Capacity Releasing Diffusion for Speed and Locality}

\begin{icmlauthorlist}
\icmlauthor{Di Wang}{eecs}
\icmlauthor{Kimon Fountoulakis}{stat}
\icmlauthor{Monika Henzinger}{uvien}
\icmlauthor{Michael W. Mahoney}{stat}
\icmlauthor{Satish Rao}{eecs}
\end{icmlauthorlist}

\icmlaffiliation{eecs}{EECS, UC Berkeley, Berkeley, CA, USA}
\icmlaffiliation{stat}{ICSI and Statistics, UC Berkeley, Berkeley, CA, USA}
\icmlaffiliation{uvien}{Computer Science, University of Vienna, Vienna, Austria}
\icmlcorrespondingauthor{Di Wang}{wangd@eecs.berkeley.edu}

\vskip 0.3in
]
\printAffiliationsAndNotice{}
\begin{abstract} 
\input{abstract}
\end{abstract}
\input{intro}

\input{crd}

\input{local_cluster}

\input{empirical}

\section*{Acknowledgements} 
SR and DW are supported by the National Science Foundation under Grant CCF-1528174 and CCF-1535989. MM and KF would like to thank the 
Army Research Office 
and the
Defense Advanced Research Projects Agency
for partial support of this work. MH has received funding from the European Research Council under the European Union's Seventh Framework Programme (FP/2007-2013)/ERC Grant Agreement no. 340506.
\bibliography{local_cut}
\bibliographystyle{icml2017}
\begin{appendices}

\input{crd_app}
\input{local_cluster_app}
\input{appendix_empirical}
\end{appendices}
\end{document}

%% file: abstract.tex
Diffusions and related random walk procedures are of central importance in many areas of machine learning, data analysis, and applied mathematics.
Because they spread mass agnostically at each step in an iterative manner, they can sometimes spread mass ``too aggressively,'' thereby failing to find the ``right'' clusters.
We introduce a novel \emph{Capacity Releasing Diffusion (CRD) Process}, which 
is both faster and stays more local than the classical spectral diffusion process. 
As an application, we use our CRD Process to develop an improved local algorithm for graph clustering. Our local graph clustering method can find local clusters in a model of clustering where one begins the CRD Process in a cluster whose vertices are connected better internally than externally by an $O(\log^2 n)$ factor, where $n$ is the number of nodes in the cluster. 
Thus, our CRD Process is the first local graph clustering algorithm that is not subject to the well-known quadratic Cheeger barrier. 
Our result requires a certain smoothness condition, which we expect to be an artifact of our analysis. 
Our empirical evaluation demonstrates improved results, in particular for realistic social graphs where there are moderately good---but not very good---clusters.

%% file: intro.tex
\vspace{-7mm}
\section{Introduction}
\tikzset{
dot/.style={circle,draw=black,thick,inner sep=2pt,fill},
}

Diffusions and related random walk procedures are of central
importance in many areas of machine learning, data analysis, and
applied mathematics, perhaps most conspicuously in the area of
spectral
clustering~\citep{Cheeger69,DH73,luxburg05_survey,ShiMalik00_NCut},
community detection in
networks~\citep{NJW01_spectral,WS05_spectralSDM,LLDM09_communities_IM,Jeub15},
so-called manifold learning~\citep{BN03,MOV12_JMLR}, and PageRank-based
spectral ranking in web ranking~\citep{PB99,Gle15_SIREV}.  Particularly
relevant for our results are local/personalized versions
of PageRank \citep{JW03} and local/distributed versions of spectral
clustering \citep{ST04,ACL06,AP09}. These latter algorithms
can be used to find provably-good small-sized clusters in very large graphs
without even touching the entire
graph;   
they have been implemented and
applied to billion-node graphs~\citep{SRFM16_VLDB}; and they have been
used to characterize the clustering and community structure in a wide
range of social and information
networks~\citep{LLDM09_communities_IM,Jeub15}.


Somewhat more formally, we will use the term \emph{diffusion} on a graph to refer to a process that spreads mass among vertices by sending mass along edges step by step according to some rule.  
With this interpretation, classical \emph{spectral diffusion} spreads mass by distributing the mass on a given node equally to the neighbors of that node in an iterative manner.
A well-known problem with spectral methods is that---due to their close relationship with random walks---they sometimes spread mass ``too aggressively,'' and thereby they don't find the ``right'' partition.  
In theory, this can be seen with so-called
Cockroach Graph~\citep{guatterymiller98,luxburg05_survey}.  In
practice, this is seen by the extreme sensitivity of spectral methods
to high-degree nodes and other structural heterogeneities in
real-world graphs constructed from very noisy
data~\citep{LLDM09_communities_IM,Jeub15}.  More generally, it is
well-known that spectral methods can be very sensitive to a small
number of random edges, e.g., in small-world graphs, that ``short
circuit'' very distant parts of the original graph, as well as other noise
properties in realistic data.  Empirically, this is well-known to be a particular problem when
there are moderately good---but not very good---clusters in the data,
a situation that is all too common in machine learning and data
analysis applications~\citep{Jeub15}.

Here, we introduce a novel \emph{Capacity Releasing Diffusion (CRD) Process} to address this problem.  
Our CRD Process is a type of diffusion that spreads mass according to a carefully-constructed push-relabel rule, using techniques that are well-known from flow-based graph algorithms, but modified here to release the capacity of edges to transmit mass.  
Our CRD Process has better properties with respect to limiting the spread of mass inside local well-connected clusters.  It does so with improved running time properties.
We show that this yields improved local clustering algorithms, both in worst-case theory and in empirical practice.

\vspace{-3mm}
\subsection{Capacity Releasing Diffusion (CRD)}




We start by describing the {\em generic} CRD process in Figures~\ref{fig:two} and~\ref{fig:one}, which lays down the dynamics of spreading mass across the graph. 
Importantly, this dynamical process is independent of any particular task to which it may be applied. Later (in Section~\ref{sxn:flow}) we also present a concrete CRD algorithm for the specific task of local clustering that exploits the dynamics of the generic CRD process\footnote{The relation between the generic CRD process and the CRD algorithm for local graph clustering is analogous to the relation between local random walks and a local spectral graph partitioning algorithm such as that of \citet{ACL06}.}.
\begin{figure}[h]
\fbox{\parbox{.47\textwidth}{
{\small
\begin{enumerate}
\item
Begin with $2d(u)$ mass at a single (given) vertex $u$. 
\item
Repeatedly perform a {CRD step}, and then double the mass at every vertex.
\end{enumerate}
}
}}
\caption{{\bf Generic Capacity Releasing Diffusion (CRD) Process}}
\label{fig:two}
\end{figure}
\vskip-.15in

The entire CRD process (Figure~\ref{fig:two}) repeatedly applies the generic CRD inner process (which we call a {\em CRD step}), and then it doubles the amount of mass at all vertices between invocations. A CRD step starts with each vertex $u$ having mass $m(u)\leq 2d(u)$, where $d(u)$ is the degree of $u$, and spreads the mass so that at the end each vertex $u$ has mass $m(u)\leq d(u)$.
\begin{figure}[h]
\fbox{\parbox{.46\textwidth}{
{\small
Each vertex $v$ initially has mass $m(v)\leq 2d(v)$ and needs to spread the mass so that $m(v) \leq d(v) \forall v$.

\begin{enumerate}
\item
Each vertex $v$ maintains a label, $l(v)$, initially set to~$0$. 
\item
Each edge $e$ maintains $m(e)$, which is the mass moved from $u$ to $v$ ,where $m(v,u) = -m(u,v)$.
We note that both $l(v)$ and $m(u,v)$ are variables local to this inner process, where
$m(v)$ evolves across calls to this process.
\item
An edge, $e = (u,v)$, is {\em eligible} with respect to the labeling if it is {\em downhill}: $l(u) > l(v)$, and the mass $m(u,v)$ moved along $(u,v)$ is less than $l(u)$. 
\item
The {\em excess} of a vertex is $\ex(v)\myeq\max(0,m(v)-d(v))$.
\end{enumerate}
The inner process continues as long as there is a vertex $v$ with $\ex(v)>0$; it either sends mass over an eligible edge incident to $v$, or if there is none, then it increases $l(v)$ by $1$. 
}
}}
\caption{{\bf Generic CRD Inner Process.}}
\label{fig:one}
\end{figure}
Observe that, essentially, each CRD step spreads the mass to a region of roughly twice the volume comparing to the previous step.
%

The generic CRD inner process (Figure~\ref{fig:one}) implements a modification of the  classic ``push-relabel'' algorithm \citep{GT88,GT14} for routing a source-sink flow.
The crucial property of our process (different from the standard push-relabel) 
is that edge capacity is made available to the process slowly
by \emph{releasing}.  That is, we only allow $l(u)$ units of mass
 to move across any edge $(u,v)$, where $l(u)$ is the label (or height) maintained by the CRD inner process. 
Thus, edge capacity is {\em released} to
allow mass to cross the edge as
the label of the endpoint rises.  
As we will see, this difference is critical to the theoretical and empirical performance of the CRD algorithm.


\vspace{-3mm}
\subsection{Example: Classical Versus Capacity Releasing}



To give insight into the differences between classical spectral diffusion and our CRD Process, consider the graph in Figure~\ref{fig:path}.  
There is a ``cluster'' $B$, which consists of $k$ paths, each of length $l$, joined at a common node $u$. 
There is one edge from $u$ to the rest of the graph, and we assume the other endpoint $v$ has very high degree such that the vast majority of the mass arriving there is absorbed by its neighbors in $\overline{B}$.
While idealized, such an example is not completely unrealistic~\citep{LLDM09_communities_IM,Jeub15}.

\begin{figure}[ht]
\begin{tikzpicture}[yscale=.5]

\node[dot,label={{\bf $u$}}] (u) at (0,0) {};
\node[dot,label={{\bf $v$}}] (v) at (-1,0) {};

\foreach \u/\x/\y in {t11/1/1,t12/2/1,t15/4/1,m11/1/0,m12/2/0,m15/4/0,b11/1/-1,b12/2/-1,b15/4/-1,}
{
\node[dot] (\u) at (\x,\y) {};
}

\foreach \u/\x/\y in {tdots/3/1,mdots/3/0,bdots/3/-1}
{
\node (\u) at (\x,\y) {{\LARGE $\cdots$}};
}

\foreach \u/\x/\y in {dts1/1.5/.6,dts2/1.5/-.4,dts3/-1.5/.3}
{
\node (\u) at (\x,\y) {{\LARGE $\vdots$}};
}

\foreach \u/\x/\y in {l1/-1.5/1,l2/-1.5/-.5/,l3/-1.5/-1}
{
\node (\u) at (\x,\y) {};
}

\foreach \u/\v in {t11/t12,m11/m12,b11/b12,u/t11,u/m11,u/b11,u/v,v/l1,v/l2,v/l3}
{
\path[thick] (\u) edge (\v);
}

\draw [decoration={brace,raise=5pt},decorate,line width=1pt]
   (1,1.1) -- (4,1.1);

\node at (3,2) {$\ell$};

\draw [decoration={brace,raise=5pt},decorate,line width=1pt]
   (4.2,1) -- (4.2,-1);

\draw (0,2) arc(150:210:4);

\node at (.5,2) {$B$};
\node at (-1,2) {$\overline{B}$};

\node at (4.8,0) {$k$};

\end{tikzpicture}
\caption{An example where Capacity Releasing Diffusion beats classical spectral diffusion by $\Omega(\ell)$.}
\label{fig:path}
\end{figure}
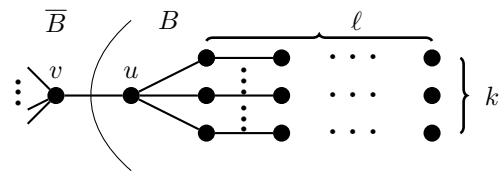

Consider first classical spectral diffusion, with a random walk starting from some vertex in $B$. 
This process requires $\Omega(\ell^2)$ steps to
spread probability mass to a constant fraction of the nodes on the paths, and
in this many steps, the expected number of times to visit $u$ is
$\Omega(\ell)$. Because of the edge to $v$, each time we visit $u$, we
have a $\Omega(1/k)$ chance of leaving $B$. Thus, when $\ell$ is
$\Omega(k)$, the random walk is expected to leave $B$ and never return,
i.e., the classical diffusion will leak out all the probability
mass before even spreading beyond a constant fraction of $B$.


Consider next our CRD Process, starting with mass at the vertex $u \in B$ (which would be a worst-case
starting node in $B$ for CRD). 
Assume that at some point the mass is spread along $z$ neighboring vertices on each of the $k$ paths.  To
continue the spread to $2z$ vertices in the next CRD step, the labels will be raised to
(at most) $2z$ to allow the mass to spread over the path of length $2z$.  This
enables the spread along the paths, but it only releases a capacity of $2z$ to the exiting edge $(u,v)$. Since in this call, a total of $2zk$
mass is in the set $B$,  at most $1/k$ of the mass escapes.  After 
$\log \ell$ CRD steps, the mass is spread over all the $k$ length-$\ell$ 
paths, and only a  $(2\log \ell)/k$ fraction of the mass has
escaped from $B$.  Thus if $\ell = \Omega(k)$, as before, 
a factor of $k/\log \ell$ less mass has escaped from $B$ with the CRD Process than with
the classical diffusion process.

Without the releasing, however, the mass escaping $B$ would be large, as even raising the label of vertex $u$ to 1 would allow an arbitrary amount of mass to leak out. 

Finally, note that the $\ell^2$ mixing time 
 makes spectral
diffusions a $\Omega(\ell)$ factor slower than CRD. This drawback
 of spectral techniques
 can
perhaps be resolved using sophisticated methods such as
 evolving sets~\citep{AP09}, though it comes easily with
CRD. 

\vspace{-3mm}
\subsection{Our Main Results}

We provide theoretical and empirical evidence that our CRD algorithm is superior to classical diffusion methods at finding clusters, with respect to noise tolerance, recovery accuracy, cut conductance, and running time. 
Here the {\em cut conductance} $\Phi(S)$ of a cut $(S, \bar S)$
is $\Phi(S):= \frac{|E(S,\bar{S})|}{\min(\mbox{vol}(S),\mbox{vol}(V\setminus
  S))}$, where $E(S,\bar{S})$ denotes the set of edges between $S$ and
$\bar{S}$,\footnote{Unless otherwise noted, when speaking of the conductance of
a cut $S$, we assume $S$ to be the side of minimum volume.} and the {\em volume} $\vol(S)$ is the
sum of the degrees of the vertices in $S$. 
 In all these measures, we break the quadratic Cheeger barrier  for
classical diffusions (explained below) while 
presenting a {\em local} algorithm (i.e., an algorithm whose
running time depends on the volume of the cluster found and {\em not} of the whole~graph).

Our first main result (Section~\ref{sxn:flow}) presents a CRD algorithm and its running time. The CRD algorithm is a parameterized specialization of the generic CRD Process, where we limit the maximum label of vertices, as well as the maximum edge capacity. We prove that this specialization is efficient, in that it runs in time linear in the total mass and the label limit, and it either succeeds in spreading the mass or it leaves all unspread mass at  nodes with  high label.  
This property is analogous to the ispoerimetric capacity control provided by local spectral methods, and it is important for locating cluster bottlenecks.  We use this crucially in our context to find low conductance clusters.

Our second main result (Section~\ref{sxn:local_cluster}) concerns the use of the CRD algorithm to find good local clusters in large graphs.
Our result posits the existence of a ``good'' cluster $B$, which satisfies certain conditions (Assumption \ref{assum:bottleneck} and \ref{assum:connectmore}) that naturally capture the notion of a local structure. 
The rather weak Assumption \ref{assum:bottleneck} states that $B$'s internal connectivity $\phi_S(B)$ (see Section~\ref{sxn:local_cluster} for definition) is
a constant factor better (i.e., larger) than the conductance $\phi(B)$. 
Assumption \ref{assum:connectmore} states that we have a
smoothness condition which needs
that any subset $T \subset B$ has $\polylog(\vol(B))$ times more neighbors in $B-T$
than in $V-B$. 
Under these conditions, we can recover $B$ starting from any vertex in~$B$.

Both assumptions formalize the idea that the signal of the local structure is
 stronger than the noise of the cluster by some moderately large factor. 
More specifically, Assumption \ref{assum:bottleneck} roughly says that the weakest signal of any subset of $B$ is
a constant times stronger than the average noise of $B$; and Assumption \ref{assum:connectmore} roughly says the signal of any subset is $\polylog(\vol(B))$ times stronger than the noise of the subset.

We note that Assumption \ref{assum:bottleneck} is significantly weaker than
the factor in \citet{ALM13}, where it is shown how to  localize a 
cluster $B$ such that  $\phi_S(B) \ge \sqrt{\phi(B)}$. Their condition is considerably
stricter than our condition on the ratio between $\phi_S(B)$ and $\phi(B)$,
especially when $\phi(B)$ is small, as is common. Their
algorithm relies on proving that a classical diffusion starting at a
typical node keeps most of its mass inside of $B$. However, they do not need
something like our smoothness condition.

With the additional smoothness condition, we break the dependence on
$\sqrt{\phi(B)}$ that is central to all approaches using spectral
diffusions, including \citet{ALM13}, for the first time with a local
algorithm.  In particular, comparing to \citet{ALM13}, under their
parameter settings (but with the smoothness condition), we identify
a cluster with $\sqrt{\phi(B)}$ times less error, and we have a
$1/\sqrt{\phi(B)}$ speedup in running time. This improvement is (up to a $\log \ell$-factor) consistent with
the behavior in the example of the previous section where the improvement
is $k/\log \ell =  O(1/(\sqrt{1/\phi(B)} \log \ell))$ as  $\phi(B) = 1/k\ell$ and $\ell = \Omega(k)$.

We note that with the additional smoothness condition, our theoretical results hold for any starting node $v_s$ in $B$
, in contrast to prior spectral-based results which only work when starting from a ``good'' node (where only a constant fraction of the nodes in $B$ are good). We expect the smoothness condition to be an artifact of our analysis, i.e., similar results actually hold when starting at good nodes in $B$, even without this assumption.

Our third main result (Section~\ref{sxn:empirical}) is an empirical illustration of our method.
We consider several social and information networks studied previously that are particularly challenging for spectral methods.
In particular, while graphs that have upward-sloping NCPs (Network Community Profiles) have good small clusters~\citep{LLDM09_communities_IM,Jeub15}, denser social networks with flat NCPs do not have any \emph{very}-good conductance clusters of any size.
They do, however, often have \emph{moderately}-good clusters, but these are very difficult for spectral methods to identify~\citep{Jeub15}.
Our empirical results show that our CRD-based local clustering algorithm is better able to identify and extract in a strongly local running time moderately good quality clusters from several such social networks.

\vspace{-3mm}
\subsection{Previous Work: Low Conductance Cuts, Diffusions, and Multicommodity Flow}


Spectral algorithms for computing eigenvalues use some variant of repeated matrix multiplication, which for graphs is a type of classical diffusion. 
For the Laplacian of a graph, the convergence rate is $O(1/\lambda_2)$, where $\lambda_2$ is the second smallest eigenvalue by of this matrix.  
The Lanczos method improves this rate to $O(\sqrt{1/\lambda_2})$  by cleverly and efficiently combining different iterations of the diffusions. See, e.g., \citet{OSV12} for more details on this.

One application of such a computation is to find a low conductance cut in a graph.
The second eigenvector for $G$ can  be used to find a cut of conductance $O(\frac{1}{\lambda_2})$ \citep{Cheeger69,DH73}. 
Let $\phi_G$ be the minimum conductance in the graph.
In his work, Cheeger already observed that random-walk based diffusion can make a $\Theta(1/\sqrt{\phi_G})$ error in estimating the conductance, informally known as the (quadratic) {\em Cheeger barrier}, and illustrated in our example.
This, combined with the fact that $\lambda_2 = O(1/\phi_G)$, gives a spectral method to find an $O(\phi^{1/2}_G)$ conductance cut in $G$.

Spielman-Teng \yrcite{ST04} used local versions of diffusions (i.e., those with small support) to compute recursive decompositions efficiently, and then they used locality to produce linear time partitioning algorithms. 
Andersen, Chung and Lang \yrcite{ACL06} developed an improved version that
adjusts the standard diffusion by having mass settled at vertices,
resulting in significantly improved bounds to $O(\sqrt{\phi_G \log n})$
on the conductance of the returned cut $(B, \bar B)$
in time $\tilde{O}(\frac{\vol(B)}{\phi_G})$.
Allen-Zhu, Lattanzi and Mirrokni \yrcite{ALM13} analyzed the behavior of the same algorithm under certain well-connected conditions.
 The EvoCut algorithm of Andersen and Peres
\yrcite{AP09} improved the running time of this method to 
$\tilde{O}(\frac{\vol(B)}{\sqrt{\phi_G}}).$
As all these methods are based on spectral diffusion, their performance with
respect to conductance is subject to the Cheeger barrier.
%
Other processes have been proposed for random
walks that mix faster, e.g., non-backtracking random walks
\citep{Alon07}.  These too are subject to the Cheeger barrier
asymptotically.
Our result is the first to break this barrier in any broad setting, where classical spectral methods fail.

There is also a line of research which gives various generalizations of Cheeger's inequality \citep{LRTV12,LGT12,KLLGT13,KLL16}. On a high level, these results prove better quality for the cut or $k$-way cut produced by spectral methods by looking beyond the second eigenvalue of the graph Laplacian. In particular, the cut produced by spectral methods can have better quality than what the standard Cheeger's bound implies when the gap between the $k$-th eigenvalue and the previous eigenvalues is large. On a high level, such gap suggests the existence of a partition of the graph into $O(k)$ clusters that are much better connected internally than the cuts along the clusters, and our work qualitatively agrees with these results in the local setting. We note that these results also suffer a quadratic loss, although in terms of the $k$-th eigenvalue instead of the second eigenvalue.

Multicommodity flow based methods are able to find clusters of conductance $O(\phi_G\log n)$ \citep{LR88}, bypassing the limit inherent in purely spectral methods.
A semidefinite programming approach, which can be viewed as combining multicommodity flow and spectral methods, yields cuts of conductance $O(\phi_G\sqrt{\log n})$ \citep{ARV04}.
These algorithms are very non-local, e.g., in the sense that their running time depends on the size of the whole graph, and it is not clear that they can be meaningfully localized.
%
%
%
%
%
We do, however, use well-known flow-based ideas in our algorithm.
In particular, recall that push-relabel and in general ``shortest-path'' based methods have a celebrated history in algorithms \citep{GT14}.  Using levels to release capacity, however, as we do in our algorithm, is (to our knowledge) completely new.



%% file: crd.tex
\newcommand{\first}{\operatorname{first}}
\newcommand{\last}{\operatorname{last}}
\newcommand{\current}{\operatorname{current}}
\vspace{-3mm}
\section{Capacity Releasing Diffusion}
\label{sxn:flow}
In this section, we describe our algorithm which implements a specific version of the generic CRD Process.
In particular, it has some modifications for efficiency reasons, and it terminates the diffusion when it finds a bottleneck during the process. The algorithm iteratively calls a subroutine {\em CRD-inner}, which implements one CRD step. 

For efficiency reasons, {\em CRD-inner} doesn't necessarily carry out a full CRD step, where a {\em  full CRD step} means every node $u$ has at most $d(u)$ mass at termination. In particular, {\em CRD-inner} only makes a certain amount of ``effort'' (which is tuned by a parameter $\phi$) to spread the mass, and if there is a  bottleneck in the form of
a cut that requires ``too much effort'' for the diffusion to get through, then {\em CRD-inner} may leave excess mass on nodes, i.e., $m(v)>d(v)$ at termination. More specifically, given $\phi$, {\em CRD-inner} guarantees to overcome any  bottleneck
 of conductance $\Omega(\phi)$, i.e., if it doesn't carry out a full CRD step, then it returns a cut of conductance $O(\phi)$ as a certificate. We will discuss {\em CRD-inner} with more detail in Section~\ref{sxn:crd-inner}.
 
\vspace{-3mm}
 \subsection{CRD Algorithm}

Given a starting node $v_s$, the CRD algorithm (Algorithm~\ref{alg:CRD}) is essentially the CRD Process starting from $v_s$, as described in Figure~\ref{fig:two}. The algorithm takes as input a parameter $\phi$, which is used to tune {\em CRD-inner}. Since {\em CRD-inner} may stop short of a full CRD step due to a bottleneck, we remove any excess mass remaining on nodes after calling {\em CRD-inner}. Due to the excess removal, we may discard mass as the algorithm proceeds. In particular, as we start with $2d(v_s)$ mass, and double the amount after every CRD step, the amount of mass 
after the $j$-th doubling is $2d(v_s)\cdot 2^j$ if we never remove excess. 
When the actual amount of mass is significantly smaller than $2d(v_s)\cdot 2^j$, there must be a bottleneck $(K, \bar K)$ during the last CRD step, such that $K$ contains
a large fraction of the mass (and of the excess) and such that {\em CRD-inner} cannot push any more
mass from $K$ to $\bar K$. We terminate the CRD algorithm when this happens, as 
the mass and, as we can show, thus the volume of $K$ must be large, while there are few edges
between $K$ and $\bar K$. Thus $K$ is a  low-conductance cluster around $v_s$. Formally, the algorithm takes input parameters $\tau$ and $t$, and it terminates either when the amount of mass drops below $\tau(2d(v_s)\cdot 2^j)$ after iteration $j$, or after iteration $t$ if the former never happens. It returns the mass on the  nodes (i.e., $m(\cdot)$), as well as the cut $K$ returned by the last {\em CRD-inner} call in the former termination state.  

The running time of our CRD algorithm is local (i.e., proportional to the volume of the region it spreads mass to, rather than the volume of the entire graph). In particular, each {\em CRD-inner} call takes time linear in the amount of mass, 
 and as the amount of mass increases geometrically before we terminate, the running time of the CRD algorithm is dominated by the last {\em CRD-inner} call.
\noindent
\begin{algorithm}[t]
{\small
\caption{{\em CRD Algorithm}($G,v_s,\phi,\tau,t$)}
\label{alg:CRD}
~\\
\fbox{
\parbox{0.47\textwidth}{
\tab {\bf Initialization}: \\
\phantom{\tab} $m(v_s)=d(v_s),m(v)=0,\forall v\neq v_s$; $j=0$.\\
\tab {\bf For} $j=0,\ldots, t$\\
\tab \tab $m(v)\leftarrow 2m(v), \forall v$ \\
\tab \tab {\bf Assertion:} $m(v)\leq 2d(v), \forall v$\\
\tab \tab Call {\em CRD-inner} with $G,m(\cdot),\phi$, get cut $K_j$\\
\phantom{\tab \tab} ($K_j$ empty if CRD-inner finishes full CRD step).\\
\tab \tab $m(v)\leftarrow \min(m(v),d(v)), \forall v$\\
\tab \tab {\bf If } $\sum_v m(v)\leq \tau(2d(v_s)\cdot 2^j)$ \\ 
\tab \tab \tab {\bf Return} $m(\cdot)$, and $K\myeq K_j$. {\bf Terminate.}\\
\tab {\bf End For}\\ 
\tab {\bf Return} $m(\cdot), K\myeq K_t$.
}}
}
\end{algorithm}
\input{crd-inner.tex}

%% file: crd-inner.tex
\vspace{-3mm}
\subsection{CRD Inner Procedure}
\label{sxn:crd-inner}
Now we discuss the {\em CRD-inner} subroutine (Algorithm~\ref{alg:CRD-inner}), which aims to carry out one CRD step. In particular, each node $v$ has $m(v)\leq 2d(v)$ mass at the beginning, and {\em CRD-inner} tries to spread the mass so each node $v$ has $m(v)\leq d(v)$ mass at the end. Not surprisingly, as the CRD step draws intuition from flow routing, our {\em CRD-inner} can be viewed as a modification of the classic push-relabel algorithm. 

As described in Figure~\ref{fig:one}, we maintain a label $l(v)$ for each node $v$, and the net mass being pushed along each edge. Although the graph is undirected, we consider each edge $e=\{u,v\}$ as two directed arcs $(u,v)$ and $(v,u)$, and we use $m(u,v)$ to denote the net mass pushed from $u$ to $v$ (during the current {\em CRD-inner} invocation). Under this notation, we have $m(u,v)=-m(v,u)$. We denote $|m(\cdot)|\myeq \sum_v m(v)$ as the total amount of mass, $\ex(v)\myeq \max(m(v)-d(v),0)$ as the amount of excess on $v$, and we let $\phi$ be the input parameter tuning the ``effort'' made by {\em CRD-inner} (which will be clear shortly). 

As noted earlier, to make {\em CRD-inner} efficient, we deviate from the generic CRD step. In particular, we make the following modifications:
\begin{enumerate}[topsep=0pt,itemsep=-0.5ex,partopsep=1ex,parsep=1ex]
\item The label of any node can be at most $h=3\log |m(\cdot)|/\phi$. If $v$ is raised to level $h$, but still has excess mass, {\em CRD-inner} leaves the excess on $v$, and won't work on $v$ any more. Formally, $v$ is {\em active} if $l(v)<h$ and $\ex(v)>0$. We keep a list $Q$ of all active nodes, and terminate {\em CRD-inner} when $Q$ is empty. 
\item In addition to capacity releasing, the net mass along any edge can be at most $C=1/\phi$. Formally, for an arc $(v,u)$, its {\em effective capacity} is $\hat{c}(v,u)\myeq \min(l(v),C)$, and its {\em residual capacity} is $r_m(v,u)\myeq \hat{c}(v,u) - m(v,u)$. The arc $(v,u)$ is {\em eligible} iff $l(v)>l(u)$ (i.e., downhill) and $r_m(v,u)>0$. We only push mass along eligible arcs.
\item We enforce $m(v)\leq 2d(v)$ for all $v$ through the execution. This is assumed at the start, and we never push mass to $v$ if that
would result in $m(v) > 2d(v)$.
\end{enumerate}
The parameter $\phi$ in the first two modifications limits the work done by {\em CRD-inner}, and it captures how hard {\em CRD-inner} will try to carry out the full CRD step (e.g., when $h,C$ are infinitely large, {\em CRD-inner} implements the full CRD step). Given any $\phi$, {\em CRD-inner} makes enough effort by allowing
nodes to have height up to $h$ and by using the above edge capacities to overcome bottlenecks of conductance $\Omega(\phi)$ during the diffusion process.  If it doesn't finish the full CRD step, then it returns a cut of conductance $O(\phi)$ as certificate. 

Another motivation of tuning with parameter $\phi$ is to keep the diffusion local. Since {\em CRD-inner} doesn't try to get through low-conductance bottlenecks, the diffusion tends to spread mass over well-connected region, instead of pushing mass out of a bottleneck.
This guarantees that the work performed is linear in the 
volume of the returned cluster, i.e., that it is a strongly local algorithm, since only a small fraction of mass can leak out of the cluster.

The third modification guarantees when {\em CRD-inner} terminates with a lot of excess on nodes, the excess won't be concentrated on a few nodes, as no node
can have more mass than twice its degree, and thus the cut returned must contain a large region.

\noindent 

\begin{algorithm}[t]
{\small
\caption{{\em CRD-inner}($G$,$m(\cdot)$,$\phi$) }
\label{alg:CRD-inner}
~\\
\fbox{
\parbox{0.47\textwidth}{
\tab {\bf Initialization:}\\
\tab \tab $\forall \{v,u\}\in E$, $m(u,v)=m(v,u)=0$; $\forall v$, $l(v)=0$\\
\tab \tab $Q=\{v|m(v)>d(v)\}$, $h=\frac{3\log|m(\cdot)|}{\phi}$\\
\tab {\bf While} $Q$ is not empty\\
\tab \tab Let $v$ be the lowest labeled node in $Q$.\\
\tab \tab {\em Push/Relabel}$(v)$. \\
\tab \tab {\bf If} {\em Push/Relabel}$(v)$ pushes mass along $(v,u)$\\
\tab \tab \tab {\bf If} $v$ becomes in-active, remove $v$ from $Q$\\
\tab \tab \tab {\bf If} $u$ becomes active, add $u$ to $Q$\\
\tab \tab {\bf Else If} {\em Push/Relabel}$(v)$ increases $l(v)$ by $1$\\
\tab \tab \tab {\bf If} $l(v)=h$, remove $v$ from $Q$.
}}
\fbox{
\parbox{0.47\textwidth}{
{\em Push/Relabel}$(v)$ \\
\tab {\bf If} there is any eligible arc $(v,u)$\\
\tab \tab {\em Push}$(v,u)$.\\
\tab {\bf Else}\\
\tab \tab {\em Relabel}$(v)$.
}}
\fbox{
\parbox{0.47\textwidth}{
{\em Push}$(v,u)$\\
\tab $\psi = \min\left(\ex(v),r_m(v,u),2d(u)-m(u)\right)$\\
\tab Push $\psi$ units of mass from $v$ to $u$:\\
\phantom{\tab} $m(v,u)\leftarrow m(v,u)+\psi, m(u,v)\leftarrow m(u,v)-\psi$;\\
\phantom{\tab} $m(v)\leftarrow m(v)-\psi, m(u)\leftarrow m(u)+\psi$.
}}
\fbox{
\parbox{0.47\textwidth}{
{\em Relabel}$(v)$\\
\tab $l(v)\leftarrow l(v)+1$.
}}
}
\end{algorithm}

We have the following theorem for {\em CRD-inner}.  
\begin{restatable}{theorem}{crdinner}
\label{thm:CRD-inner}
Given $G,m(\cdot)$, and $\phi\in (0,1]$, such that $|m(\cdot)|\leq \vol(G)$, and $\forall v:m(v)\leq 2d(v)$ at the start, {\em CRD-inner} terminates with one of the following cases:
\begin{enumerate}  
\item {\em CRD-inner} finishes the full CRD step: $\forall v:m(v)\leq d(v)$. 
\item There are nodes with excess, and we can find a cut $A$ 
of conductance $O(\phi)$. Moreover, $\forall v\in A: 2d(v)\geq m(v)\geq d(v)$, and $\forall v\in \bar{A}:m(v)\leq d(v)$. 
\end{enumerate}
The running time is $O(|m(\cdot)|\log (|m(\cdot)|)/\phi)$.
\end{restatable}
\vspace{-1em}
\begin{sproof} Let $l(\cdot)$ be the labels of nodes at termination. First note all nodes with excess must be on level $h$. Moreover, since we only push from a node $v$ if it has excess (i.e., $m(v)\geq d(v)$), once a node has at least $d(v)$ mass, it always has at least $d(v)$ mass. Note further that $l(v)\geq 1$ if and only if $\ex(v)>0$ at some point during the process. Thus, we know the following: $l(v)=h\Rightarrow 2d(v)\geq m(v)\geq d(v)$; $h>l(v)\geq 1\Rightarrow m(v)=d(v)$; $l(v)=0 \Rightarrow m(v)\leq d(v)$.

Let $B_i = \{v| l(v) = i\}$. Since the total amount of mass $|m(\cdot)|$ is at most the volume of the graph, if $B_0=\emptyset$ or $B_h=\emptyset$, then we have case $(1)$ of the theorem.

Otherwise, both $B_h$ and $B_0$ are non-empty. Let the {\em level cut} $S_i = \cup_{j=i}^h B_j$ be the set of nodes with label at least $i$. We have $h$ level cuts $S_h,\ldots,S_1$, where $\vol(S_h)\geq 1$, and $S_j\subseteq
S_i$ if $j>i$. The conductance of these cuts, when we go from $S_h$ down to $S_1$, lower bounds how much the volume grows from $S_h$ to $S_1$. If all these cuts have $\Omega(\phi)$ conductance, by our choice of $h$, the volume of $S_1$ will be much larger than $|m(\cdot)|$. This gives a contradiction, since any node $v\in S_1$ has $m(v)\geq d(v)$, and we don't have enough mass. It follows that at least one of the level cuts has conductance $O(\phi)$.

As to the running time, the graph $G$ is given implicitly, and we only acess the list of edges of a node when it is active. Each active node $v$ has $d(v)$ mass, and the total amount of mass is $|m(\cdot)|$, so the algorithm touches a region of volume at most $|m(\cdot)|$. Thus, the running time has linear dependence on $|m(\cdot)|$.
Using an amortization argument one can show that the total work of the subroutine (in the worst case) is  $O(|m(\cdot)| h) = O(|m(\cdot)| \log (|m(\cdot)|)/\phi)$.
\end{sproof}
There are certain details in the implementation of {\em CRD-inner} that we don't fully specify, such as how to check if $v$ has any outgoing eligible arcs, and how (and why) we pick the active node with lowest label. These details are important for the running time to be efficient, but don't change the dynamics of the diffusion process. Most of these details are standard to push-relabel framework, and we include them (as well as the detailed proof of Theorem~\ref{CRD-inner}) in Appendix~\ref{sxn:crd-appendix}.

%% file: local_cluster.tex
\vspace{-3mm}
\section{Local Graph Clustering}
\label{sxn:local_cluster}
In this section, we provide theoretical evidence that the CRD algorithm can identify a good local cluster in a large graph if there exists one around the starting node. We define
 {\em set conductance, $\phi_S(B)$} (or {\em internal connectivity}) of a set $B\subset V$ is
the minimum conductance of any cut in the induced subgraph on $B$.



Informally, for a ``good'' cluster $B$, any inside bottleneck should have larger conductance than $\phi(B)$, and nodes in $B$ should be more connected to other nodes inside $B$ than to nodes outside. We capture the intuition formally as follows. 
\vspace{-0.75em}
\begin{restatable}{assumption}{assumone}
\label{assum:bottleneck}
$
\sigma_1 \myeq \frac{\phi_S(B)}{\phi(B)} \geq \Omega(1)
$.
\end{restatable}
\vspace{-0.75em}
\begin{restatable}{assumption}{assumtwo}
\label{assum:connectmore}
There exists $\sigma_2\geq \Omega(1)$, such that any $T\subset B$ with $vol_B(T)\leq \vol_B(B)/2$ satisfies
\[
\frac{|E(T,B\setminus T)|}{|E(T,V\setminus B)|\log \vol(B)\log \frac{1}{\phi_S(B)}}\geq \sigma_2.
\] 
\end{restatable}
\vspace{-0.75em}

Following prior work in local clustering, we formulate the goal as a promise problem, where we assume the existence of an unknown target good cluster $B\subset V$ satisfying Assumption~\ref{assum:bottleneck} and ~\ref{assum:connectmore}. In the context of local clustering, we also assume $\vol(B)\leq \vol(G)/2$. Similar to prior work, we assume the knowledge of a node $v_s\in B$, and rough estimates (i.e., within constant factor) of the value of $\phi_S(B)$ and $\vol(B)$. We use the CRD algorithm with $v_s$ as the starting node, $\phi=\Theta(\phi_S(B))$, $\tau = 0.5$, and $t=\Theta(\log\frac{\vol(B)}{d(v_s)})$. With the parameters we use, the algorithm will terminate due to too much excess removed, i.e., $|m(\cdot)|\leq \tau(2d(v_s)\cdot 2^j)$ after some iteration $j$. The region where the diffusion spreads enough mass will be a good approximation of $B$. 
\begin{restatable}{theorem}{recover}
\label{thm:recover}
Starting from any $v_s\in B$, with the above parameters, when the CRD algorithm terminates, if we let $S=\{v|m(v)\geq d(v)\}$, then we have:
\begin{enumerate}
\vspace{-1.25em}
\item $\vol(S\setminus B)\leq O(\frac{1}{\sigma})\cdot \vol(B)$
\vspace{-0.5em}
\item $\vol(B\setminus S)\leq O(\frac{1}{\sigma})\cdot \vol(B)$
\end{enumerate}
\vspace{-1.25em}
where $\sigma = \min(\sigma_1,\sigma_2)\geq \Omega(1)$, with the $\sigma_1,\sigma_2$ from Assumption~\ref{assum:bottleneck} and~\ref{assum:connectmore}. The running time is $O(\frac{\vol(B)\log \vol(B)}{\phi_S(B)})$.
\end{restatable} 
The theorem states that the cluster recovered by the CRD algorithm has both good (degree weighted) precision and recall with respect to $B$; and that the stronger the ``signal'' (relative to the ``noise''), i.e., the larger $\sigma_1,\sigma_2$, the more accurate our result approximates $B$.

If the goal is to minimize conductance, then we can run one extra iteration of the CRD algorithm after termination with a smaller value for $\phi$ (not necessarily $\Theta(\phi_S(B))$ as used in previous iterations). In this case, we have the following.
\begin{theorem}
\label{thm:conductance}
If we run the CRD algorithm for one extra iteration, with $\phi \geq \Omega(\phi(B))$, then {\em CRD-inner} will end with case $(2)$ of Theorem~\ref{thm:CRD-inner}. Let $K$ be the cut returned.  We have:
\begin{enumerate}
\vspace{-1.25em}
\item $\vol(K\setminus B)\leq O(\frac{\phi(B)}{\phi})\cdot \vol(B)$
\vspace{-0.5em}
\item $\vol(B\setminus K)\leq O(\frac{\phi(B)}{\phi_S(B)})\cdot \vol(B)$
\vspace{-0.5em}
\item $\phi(K)\leq O(\phi)$
\end{enumerate}
\vspace{-1.25em}
The running time is $O(\frac{\vol(B)\log \vol(B)}{\phi})$.
\end{theorem}
Now we can search for the smallest $\phi$ that gives case $(2)$ of Theorem~\ref{thm:CRD-inner}, which must give a cut of conductance within an $O(1)$ factor of the best we can hope for (i.e., $\phi(B)$). If we search with geometrically decreasing $\phi$ values, then the running time is $O(\vol(B)\log\vol(B)/\phi(B))$.

Theorem~\ref{thm:recover} and~\ref{thm:conductance} hold due to the particular flow-based dynamics of the CRD algorithm, which tends to keep the diffusion local, without leaking mass out of a bottleneck. 

Formally, for each {CRD} step, we can bound the total amount of mass that starts on nodes in $B$, and leaves $B$ at any point during the diffusion. We have the following lemma, a sketch of the proof of which is given. We include the full proof in Appendix~\ref{sxn:clustering-appendix}.
\begin{restatable}{lemma}{leakage}
\label{lemma:leakage}
In the $j$-th {CRD} step, let $M_j$ be the total amount of mass in $B$ at the start, and let $L_j$ be the amount of mass that ever leaves $B$ during the diffusion.  Then 
$
L_j\leq O(\frac{1}{\sigma_2 \log \vol(B)})\cdot M_j
$, when $M_j\leq \vol_B(B)/2$; and $L_j\leq O(\frac{1}{\sigma_1})\cdot M_j$, when $M_j\geq \vol_B(B)/2$.
\end{restatable}
\vspace{-1em}
\begin{sproof}
We have two cases, corresponding to whether the diffusion already spread a lot of mass over $B$. 

In the first case, if $M_j \geq \vol_B(B)/2$, then we use the upper bound $1/\phi$ that is enforced on the net mass over any edge to limit the amount of mass that can leak out. In particular $L_j \leq O(\vol(B)\phi(B)/\phi_S(B))$, since there are $\vol(B)\phi(B)$ edges from $B$ to $\bar{B}$, and $\phi=\Theta(\phi_S(B))$ in {\em CRD-inner}. As $M_j \geq \Omega(\vol(B))$, we have
$L_j\leq O(\frac{1}{\sigma_1})\cdot M_j$.

The second case is when $M_j \leq \vol_B(B)/2$. In this case, a combination of Assumption \ref{assum:connectmore} and capacity releasing controls the leakage of mass. Intuitively, there are still many nodes in $B$ to which the diffusion can spread mass. For the nodes in $B$ with excess on them, when they push their excess, most of the downhill directions go to nodes inside $B$. As a consequence of capacity releasing, only a small fraction of mass will leak out.
\end{sproof}
\vspace{-1em}
Theorem~\ref{thm:recover} and~\ref{thm:conductance} follow from straightforward analysis of the total amount of leaked mass at termination. We sketch the ideas for the proof of Theorem~\ref{thm:recover}, with the full proof in Appendix~\ref{sxn:clustering-appendix}.
\vspace{-1em}
\begin{sproof}
Since we use $\phi=\Theta(\phi_S(B))$ when we call {\em CRD-inner}, the diffusion will be able to spread mass over nodes inside $B$, since there is no bottleneck with conductance smaller than $\phi_S(B)$ in $B$.

Thus, before every node $v$ in $B$ has $d(v)$ mass on it (in which case we say $v$ is {\em saturated}), there will be no excess on nodes in $B$ at the end of a {CRD} step. Consequently, the amount of mass in $B$ only decreases (compared to the supposed $2d(v_s)\cdot 2^j$ amount in the $j$-th CRD step) due to mass leaving $B$.

As long as the total amount $M_j$ of mass in $B$ at the start of a CRD step is less than $\vol_B(B)/2$,
 the mass loss to $\bar B$ is at most a $O(1/(\sigma_2 \log \vol(B)))$ fraction of the mass in $B$ each CRD step.
After $O(\log \vol(B))$ CRD steps, $M_j$ reaches $\vol(B)_B/2$, and only a $O(1/\sigma_2)$ fraction of mass
has left $B$ so far. After $O(1)$ more CRD steps,
there will be enough mass to saturate all nodes in $B$, and each of these CRD steps looses
at most a $O(1/\sigma_1)$ fraction of the mass to $\bar B$.
 Thus we loose at most a $O(1/\sigma)$ fraction of mass before all nodes in $B$ are saturated. 

Once the diffusion has saturated all nodes in $B$, the amount of mass in $B$ will be $2\vol(B)$ at the start of every subsequent {CRD} step. At most $\vol(B)\phi(B)/\phi_S(B)\leq O(\vol(B)/\sigma)$ mass can leave $B$, and nodes in $B$ can hold $\vol(B)$ mass, so there must be a lot of excess (in $B$) at the end. Thus, the CRD algorithm will terminate in at most $2$ more CRD steps, since the amount of mass almost stops growing due to excess removal.

At termination, the amount of mass is $\Theta(\vol(B))$, and only $O(1/\sigma)$ fraction of the mass is in $\bar{B}$. Since $S=\{v|m(v)\geq d(v)\}$, and the total mass outside is $O(\vol(B)/\sigma)$, we get claim $(1)$ of the theorem. In our simplified argument, all nodes in $B$ have saturated sinks (i.e., $\vol(B\setminus S)=0$) at termination. We get the small loss in claim $(2)$ when we carry out the argument in more detail.

The amount of mass grows geometrically before the CRD algorithm terminates, so the running time is dominated by the last {CRD} step. The total amount of mass is $O(\vol(B))$ in the last CRD step, and the running time follows Theorem~\ref{thm:CRD-inner} with $\phi=\Theta(\phi_S(B))$.
\end{sproof}
\vspace{-1.25em}
The proof of Theorem~\ref{thm:conductance} is very similar to  Theorem~\ref{thm:recover}, and the conductance guarantee follows directly from Theorem~\ref{thm:CRD-inner}.

%% file: empirical.tex
\vspace{-3mm}
\section{Empirical Illustration}
\label{sxn:empirical}

We have compared the performance of the CRD algorithm (Algorithm \ref{alg:CRD}), the
Andersen-Chung-Lang local spectral algorithm (ACL) \yrcite{ACL06}, and the flow-improve algorithm (FlowImp)
\citep{RL2008}.  Given a starting node $v_s$ and teleportation
probability $\alpha$, ACL is a local algorithm that computes an approximate
personalized PageRank vector, which is then used to identify local
structure via a sweep cut. FlowImp is a flow-based algorithm that takes
as input a set of reference nodes and finds a cluster around the given
reference set with small conductance value.  Note that we only couple FlowImp with
ACL.  The reason is that, while FlowImp needs a very good reference set as input to give
meaningful results in our setting, it can be used as a ``clean up'' step
for spectral methods, since they give good enough output. Note also that FlowImp has running time that depends on
the volume of the entire graph, as it optimizes a global objective,
while our CRD algorithm takes time linear in the volume of the local
region explored.

\begin{table}[t] 
    \caption{Ground truth clusters}
\vspace{2mm}
  \centering
{\scriptsize
    \begin{tabular}{ccccc}
    graph & feature & volume & nodes & cond.\\
    \midrule
     \multirow{2}{*}{\rotatebox[origin=c]{90}{Hop.}} & $217$ & $10696$ &$200$  &$0.26$ \\
     												    & $2009$& $32454$  &$886$&$0.19$\\\hdashline
     \multirow{2}{*}{\rotatebox[origin=c]{90}{Rice}} & $203$& $43321$ &$403$  & $0.46$\\
     												     & $2009$&  $30858$&$607$  & $0.33$\\\hdashline
     \multirow{2}{*}{\rotatebox[origin=c]{90}{Sim.}} & $2007$& $14424$ & $281$  & $0.47$\\
     												      & $2009$&$11845$  &$277$  & $0.1$\\\hdashline
     \multirow{2}{*}{\rotatebox[origin=c]{90}{\hspace{-0.9cm}Colgate}} & $2006$& $62064$ &$556$  & $0.48$\\
     												   & $2007$& $68381$ & $588$  & $0.41$\\
												    & $2008$& $62429$ & $640$  & $0.29$\\
												    & $2009$& $35369$ & $638$  & $0.11$\\
    \bottomrule
    \end{tabular}%
  \label{tab:clusters}%
}

\end{table}%

We compare these methods on $5$ datasets, one of which is a synthetic grid graph. 
For the $4$ real-world graphs, we use the Facebook college graphs of John Hopkins (Hop.), Rice, Simmons (Sim.), and Colgate, as introduced in \citet{TMP2012}.
 Each graph in the Facebook dataset comes along with some features, e.g., ``dorm $217$,'' and ``class year $2009$.'' 
We consider a set of nodes with the same feature as a ``ground truth'' cluster, e.g., students of year $2009$. We filter out very noisy features via some reasonable thresholds, and we run our computations on the the remaining features. We include the details of the graph and feature selection in Appendix~\ref{sxn:empirical-appendix}. 
The clusters of the features we use are shown in Table \ref{tab:clusters}.

We filter bad clusters from all the ground truth clusters, by setting reasonable thresholds on volume, conductance, and gap (which is the ratio between the spectral gap of the induced graph of cluster, and the cut conductance of the cluster). Details about the selection of the ground truth clusters are in Appendix~ref{sxn:empirical-appendix}. 
In Table \ref{tab:clusters}, we show the size and conductance of the clusters of the features used in our experiments.

For the synthetic experiment, we measure performance by conductance;
the smaller the better.  For real-world experiments, we use precision
and recall.  
We also compare to
ACLopt which ``cheats'' in the sense that {\em it uses ground
  truth} to choose the parameter $\alpha$ with best F1-score
(a combination of precision and recall).  A detailed discussion on parameter tuning of the algorithms is given in Appendix~\ref{sxn:empirical-appendix}.
%

For the synthetic data, we use a grid graph of size $60\times
60$. We add noise to the grid by randomly connecting two vertices.  We
illustrate the performance of the algorithms versus probability of 
random connection in Figure \ref{fig:grid}.  The range of
probabilities was chosen consistent with theory.
As expected, CRD outperforms ACL in the intermediate range, and the two method's performances meet at the endpoints.  One view of this is that the random
connections initially adds noise to the local structure
and eventually destroys it. CRD is more tolerant to this noise
process. 

\begin{figure}[t] 
\centering
  \includegraphics[scale=0.5]{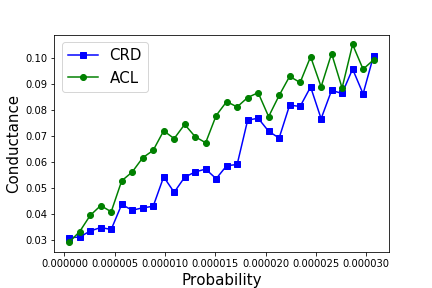}
\vspace{-2mm}
\caption{Average results for $60\times 60$ grid}
\label{fig:grid}
\end{figure}


See Table \ref{tab:results} for results for real-world data.
We run algorithms starting at each vertex in a random sample of half the vertices in each cluster and report the median. 


\begin{table}[t] 
     \caption{Results on Facebook graphs}
\vspace{2mm}
  \centering
{\scriptsize
    \begin{tabular}{ccccccc}
    ID & feat. & mtr. &CRD & ACL & ACLopt & FlowImp\\
    \midrule
     \multirow{4}{*}{\rotatebox[origin=c]{90}{Hopkins}} & \multirow{2}{*}{\rotatebox[origin=c]{0}{$217$}} & Pr. & ${0.92}$ & $0.87$&$0.87$  & ${0.92}$ \\
     								           & & Re. &${0.95}$ &$0.94$ & $0.94$& $0.89$\\
     									   & \multirow{2}{*}{\rotatebox[origin=c]{0}{$2009$}} & Pr. & $0.95$&$0.92$ & $0.91$ &${0.96}$ \\
     								           & & Re. & ${0.97}$& $0.95$& $0.95$&$0.96$\\\hdashline
     \multirow{4}{*}{\rotatebox[origin=c]{90}{Rice}} & \multirow{2}{*}{\rotatebox[origin=c]{0}{$203$}} & Pr. &${0.43}$ & $0.32$& $0.32$ &$0.33$ \\
     								           & & Re. & $0.8$& $0.9$& ${0.9}$& $0.87$\\
     									   & \multirow{2}{*}{\rotatebox[origin=c]{0}{$2009$}} & Pr. & ${0.92}$& $0.25$& $0.25$ & ${0.92}$\\
     								           & & Re. & $0.98$& $0.99$& $0.99$& ${0.99}$\\\hdashline
     \multirow{4}{*}{\rotatebox[origin=c]{90}{Simmons}} & \multirow{2}{*}{\rotatebox[origin=c]{0}{$2007$}} & Pr. &${0.5}$ &$0.49$ & $0.49$ & $0.26$ \\
     								           & & Re. & $0.5$&$0.75$ & $0.75$& ${0.99}$\\
     									   & \multirow{2}{*}{\rotatebox[origin=c]{0}{$2009$}} & Pr. & ${0.96}$& $0.95$& $0.95$ &${0.96}$ \\
     								           & & Re. & ${0.99}$ & ${0.99}$ &${0.99}$ & ${0.99}$ \\\hdashline
     \multirow{8}{*}{\rotatebox[origin=c]{90}{Colgate}} & \multirow{2}{*}{\rotatebox[origin=c]{0}{$2006$}} & Pr. &${0.43}$ &$0.41$ & $0.41$ &$0.22$ \\
     								           & & Re. &$0.53$ & $0.68$& $0.69$ & ${1.0}$\\
     									   & \multirow{2}{*}{\rotatebox[origin=c]{0}{$2007$}} & Pr. &${0.52}$ &$0.47$ & $0.47$ &$0.24$ \\
     								           & & Re. &$0.57$ &$0.71$ & $0.72$&${1.0}$\\
     									   & \multirow{2}{*}{\rotatebox[origin=c]{0}{$2008$}} & Pr. & $0.94$&$0.61$ & $0.64$ &${0.95}$ \\
     								           & & Re. & $0.96$& $0.95$&$0.95$ & ${0.97}$\\
     									   & \multirow{2}{*}{\rotatebox[origin=c]{0}{$2009$}} & Pr. &$0.97$ &$0.93$ & $0.93$ & ${0.98}$\\
     								           & & Re. & $0.98$& $0.98$& $0.98$ & ${0.99}$\\
    \midrule
    \end{tabular}%
}
  \label{tab:results}%
\end{table}%

For clusters with good but not great conductance (e.g., Rice $2009$, Colgate $2008$), CDR outperforms ACL and has
nearly identical performance to FlowImp (which, recall, is a global algorithm). This is a consequence of CDR avoiding the trap
of leaking mass out of the local structure, in contrast to ACL, which leaks a large fraction of mass. For clusters with great conductance, all methods perform very well; and all methods perform poorly when the conductance of the clusters gets close to $0.5$. We include more detailed plots and discussion in Appendix~\ref{sxn:empirical-appendix}.

Here again, as with the synthetic data, we see that for high conductance
sets (which do not have good local structure) and very good conductance
sets (which have excellent local structure), all methods perform
similarly.  In the intermediate range, i.e., when there are moderately good but not very good quality clusters, CDR shows distinct advantages, as
suggested by the theory.

%% file: crd_app.tex
\section{CRD Inner Procedure}
\label{sxn:crd-appendix}
We first fill in the missing details in the {\em CRD-inner} subroutine (Algorithm~\ref{alg:crd-inner}).  

Note an ineligible arc $(v,u)$ must remain ineligible until the next relabel of $v$, so we only need to check each arc out of $v$ once between consecutive relabels. We use $\current(v)$ to keep track of the arcs out of $v$ that we have checked since the last relabel of $v$.
\noindent 
\begin{algorithm}
\caption{{\em CRD-inner}($G$,$m(\cdot)$,$\phi$) }
\label{alg:crd-inner}
~\\
\fbox{
\parbox{0.47\textwidth}{
\tab {\bf Initialization:}\\
\tab \tab $\forall \{v,u\}\in E$, $m(u,v)=m(v,u)=0$.\\
\tab \tab $Q=\{v|m(v)>d(v)\}$, $h=\frac{3\log|m(\cdot)|}{\phi}$\\
\tab \tab $\forall v$, $l(v)=0$, and $\current(v)$ is the first edge in\\
\phantom{\tab \tab} $v$'s list of incident edges.\\
\tab {\bf While} $Q$ is not empty\\
\tab \tab Let $v$ be the lowest labeled vertex in $Q$.\\
\tab \tab {\em Push/Relabel}$(v)$. \\
\tab \tab {\bf If} {\em Push/Relabel}$(v)$ pushes mass along $(v,u)$\\
\tab \tab \tab {\bf If} $v$ becomes in-active, {\em Remove}$(v,Q)$\\
\tab \tab \tab {\bf If} $u$ becomes active, {\em Add}$(u,Q)$\\
\tab \tab {\bf Else If} {\em Push/Relabel}$(v)$ increases $l(v)$ by $1$\\
\tab \tab \tab {\bf If} $l(v)<h$, {\em Shift}$(v,Q)$\\
\tab \tab \tab {\bf Else} {\em Remove}$(v,Q)$
}}
\fbox{
\parbox{0.47\textwidth}{
{\em Push/Relabel}$(v)$ \\
\tab Let $\{v,u\}$ be $\current(v)$.\\
\tab {\bf If} arc $(v,u)$ is eligible, then {\em Push}$(v,u)$.\\
\tab {\bf Else}\\
\tab \tab {\bf If} $\{v,u\}$ is not the last edge in $v$'s list of edges.\\
\tab \tab \tab Set $\current(v)$ be the next edge of $v$.\\
\tab \tab {\bf Else} (i.e., $\{v,u\}$ is the last edge of $v$)\\
\tab \tab \tab {\em Relabel}$(v)$, and set $\current(v)$ be the first\\
\phantom{\tab \tab \tab} edge of $v$'s list of edges.
}}
\fbox{
\parbox{0.47\textwidth}{
{\em Push}$(v,u)$\\
\tab {\bf Assertion}: $r_m(v,u)>0, l(v)\geq l(u)+1$.\\
\phantom{\tab {\bf Assertion}:} $\ex(v)>0, m(u)<2d(u)$.\\
\tab $\psi = \min\left(\ex(v),r_m(v,u),2d(u)-m(u)\right)$\\
\tab Send $\psi$ units of mass from $v$ to $u$:\\
\phantom{\tab} $m(v,u)\leftarrow m(v,u)+\psi, m(u,v)\leftarrow m(u,v)-\psi$.\\
\phantom{\tab} $m(v)\leftarrow m(v)-\psi, m(u)\leftarrow m(u)+\psi$.
}}
\fbox{
\parbox{0.47\textwidth}{
{\em Relabel}$(v)$\\
\tab {\bf Assertion}: $v$ is active, and $\forall u\in V$,\\
\phantom{\tab {\bf Applicability}: }$r_m(v,u)>0\implies l(v)\leq l(u)$.\\
\tab $l(v)\leftarrow l(v)+1$.
}}
\end{algorithm}
We always pick an active vertex $v$ with the lowest label. Then for any eligible arc $(v,u)$, we know $m(u)\leq d(u)$
, so we can push at least $1$ along $(v,u)$ (without violating $m(u)\leq 2d(u)$), which is crucial to bound the total work. 

We keep the list $Q$ in non-decreasing order of the vertices' labels, for efficient look-up of the lowest labeled active vertex, and {\em Add, Remove, Shift} are the operations to maintain this order. Note these operations can be implemented to take $O(1)$ work. In particular, when we add a node $u$ to $Q$, it will always be the active node with lowest label, so will be put at the beginning. We only remove the first element $v$ from $Q$, and when we shift a node $v$ in $Q$, we know $l(v)$ increases by exactly $1$. To maintain $Q$, we simply need to pick two linked lists, one containing all the active nodes with non-decreasing labels, and another linked list containing one pointer for each label value, as long as there is some active node with that label, and the pointer contains the position of first such active node in $Q$. Maintaining this two lists together can give $O(1)$ time {\em Add, Remove, Shift}.

Now we proceed to prove the main theorem of {\em CRD-inner}.
\crdinner*
\begin{proof}
Let $l(\cdot)$ be the labels of vertices at termination, and let $B_i = \{v| l(v) = i\}$. We make the following observations: $l(v)=h\Rightarrow 2d(v)\geq m(v)\geq d(v)$; $h>l(v)\geq 1\Rightarrow m(v)=d(v)$; $l(v)=0 \Rightarrow m(v)\leq d(v)$.

Since $|m(\cdot)|\leq \vol(G)$, if $B_0=\emptyset$, it must be $|m(\cdot)|=\vol(G)$, and every $v$ has $m(v)=d(v)$, so we get case $(1)$. If $B_h = \emptyset$,  we also get case $(1)$.

If $B_h,B_0 \ne \emptyset$, let $S_i = \cup_{j=i}^h B_j$ be the set of nodes with label at least $i$. We have $h$ level cuts $S_h,\ldots,S_1$, where $\vol(S_h)\geq 1$, and $S_j\subseteq
S_i$ if $j>i$. We claim one of these level cuts must have conductance $O(\phi)$. For any $S_i$, we divide the edges from $S_i$ to $\overline{S_i}$ into two groups: $1)$ edge across one level (i.e., from node in $B_i$ to node in $B_{i-1}$), and $2)$ edges across more than one level. Let $z_1(i),z_2(i)$ be the number of edges in the two groups respectively, and define $\phi_g(i)\myeq z_g(i)/\vol(S_i)$ for $g=1,2$. 

First we show that, there must be a $i^*$ between $h$ and $h/2$ such that $\phi_1(i^*)\leq \phi$. By contradiction, if $\phi_1(i)>\phi$ for all $i=h,\ldots,h/2$, since $\vol(S_{i-1})\geq \vol(S_i)(1+\phi_1(S_i))$, we get $\vol(S_{h/2})\geq (1+\phi)^{h/2}\vol(S_h)$. With $h=3\log|m(\cdot)|/\phi$, we have $\vol(S_{h/2})\geq \Omega(|m(\cdot)|^{3/2})$, and since nodes in $S_{h/2}$ are all saturated, we get a contradiction since we must have $\vol(S_{h/2})\leq |m(\cdot)|$.

Now we consider any edge $\{v,u\}$ counted in $z_2(i^*)$ (i.e., $v\in S_{i^*},u\in\overline{S_{i^*}},l(v)-l(u)\geq 2$). Since $i^*\geq h/2 > 1/\phi$, $\hat{c}(v,u)=1/\phi$. $l(v)-l(u)>2$ suggests $r_m(v,u)=0$, thus $m(v,u)=1/\phi$ (i.e., $1/\phi$ mass pushed out of $S_{i^*}$ along each edge counted in $z_2(i^*)$). Each edge counted in $z_1(i^*)$ can have at most $1/\phi$ mass pushed into $S_{i^*}$, and at most $2\vol(S_{i^*})$ mass can start in $S_{i^*}$, then we know
\[
z_2(i^*)/\phi \leq z_1(i^*)/\phi + 2 \vol(S_{i^*})
\]
We will let $A$ be $S_{i^*}$, and we have 
\[
\phi(A)=\frac{z_1(i^*)+z_2(i^*)}{\vol(S_{i^*})}\leq 4\phi = O(\phi)
\]
Here we assume $S_{i^*}$ is the smaller side of the cut to compute the conductance. If this is not the case, i.e. $\vol(S_{i^*})> \vol(G)/2$, we just carry out the same argument as above, but run the region growing argument from level $h/4$ up to level $h/2$, and get a low conductance cut, and still let $A$ to be the side containing $S_h$.
The additional properties of elements in $A$ follows from $S_h\subseteq A\subseteq S_{h/4}$.

Now we proceed to the running time. The initialization takes $O(|m(\cdot)|)$. Subsequently, each iteration takes $O(1)$ work. We will first attribute the work in each iteration to either a push or a relabel. Then we will charge the work on pushes and relabels to the absorbed mass, such that each unit of absorbed mass gets charged $O(h)$ work. Recall the absorbed mass at $v$ are the first up to $d(v)$ mass starting at or pushed into $v$, and these mass never leave $v$, as the algorithm only pushes excess mass. This will prove the result, as there are at most $|m(\cdot)|$ units of (absorbed) mass in total.

In each iteration of {\em Unit-Flow}, the algorithm picks a lowest labeled active node $v$. If {\em Push/Relabel}$(v)$ ends with a push of $\psi$ mass, we charge $O(\psi)$ to that push operation. Since $\psi\geq 1$, we charged the push enough to cover the work in that iteration. If the call to {\em Push/Relabel}$(v)$ doesn't push, we charge the $O(1)$ work of the iteration to the next relabel of $v$ (or the last relabel if there is no next relabel). 
The latter can happen at most $d(v)$ times between consecutive relabels of $v$, so each relabel of $v$ is charged $O(d(v))$ work.

We now charge the work on pushes and relabels to the absorbed mass. Note each time we relabel $v$, there are $d(v)$ units of absorbed mass at $v$, so we charge the $O(d(v))$ work on the relabel to the absorbed mass, and each unit gets charged $O(1)$. There is at most $h$ relabels of $v$, so each unit of absorbed mass is charged $O(h)$ in total by all the relabels.

For the work on pushes, we consider the potential function $\Lambda=\sum_v \ex(v)l(v)$.
$\Lambda$ is always non-negative, and as we only push excess mass downhill, each push of $\psi$ units of mass decrease $\Lambda$ by at least $\psi$, so we can charge the work on pushes to the increment of $\Lambda$. It only increases at relabel. When we relabel $v$, $\Lambda$ is increased by $\ex(v)$. Since $\ex(v)\leq d(v)$, we can charge $O(1)$ to each unit of absorbed mass at $v$ to cover $\Lambda$'s increment. In total we can charge all pushes (via $\Lambda$) to absorbed mass, and each unit is charged with $O(h)$. 

If we need to compute the cut $A$ in case $(2)$, the running time is $O(\vol(S_1))$, which is $O(|m(\cdot)|)$.
\end{proof}

%% file: local_cluster_app.tex
\section{Local Clustering}
\label{sxn:clustering-appendix}
Recall we assume $B$ to satisfy the following conditions.
\assumone*
\vspace{-0.5em}
\assumtwo*
Now we proceed to prove the main lemma.
\leakage*
\vspace{-1.5em}
\begin{proof}
For simplicity, we assume once a unit of mass leaves $B$, it is never routed back. Intuitively, mass coming back into $B$ should only help the algorithm, and indeed the results don't change without this assumption. We denote $|M_j(S)|$ as the amount of mass on nodes in a set $S$ at the start of the {\em CRD-inner} call.

We have two cases, corresponding to whether the diffusion already spread a lot of mass over $B$. If $M_j \geq \vol_B(B)/2$, we use the upperbound $1/\phi$ that is enforced on the net mass over any edge to limit the amount of mass that can leak out. In particular $L_j \leq O(\vol(B)\phi(B)/\phi_S(B))$, since there are $\vol(B)\phi(B)$ edges from $B$ to $\bar{B}$, and $\phi=\Theta(\phi_S(B))$ in {\em CRD-inner}. As $M_j \geq \Omega(\vol(B))$, we have
$L_j\leq O(\frac{1}{\sigma_1})\cdot M_j$.

The second case is when $M_j \leq \vol_B(B)/2$. In this case, a combination of Assumption \ref{assum:connectmore} and capacity releasing controls the leakage of mass. Intuitively, there are still many nodes in $B$ that the diffusion can spread mass to. For the nodes in $B$ with excess on them, when they push their excess, most of the downhill directions go to nodes inside $B$. As a consequence of capacity releasing, only a small fraction of mass will leak out.

In particular, let $l(\cdot)$ be the labels on nodes when {\em CRD-inner} finishes, we consider $B_i=\{v\in B|l(v)=i\}$ and the level cuts $S_i=\{v\in B|l(v)\geq i\}$ for $i=h,\ldots,1$. As $M_j\leq \vol_B(B)/2$, we know $\vol(S_h)\leq \vol(S_{h-1})\leq \ldots\leq \vol(S_1)\leq \vol_B(B)/2$. In this case, we can use Assumption~\ref{assum:connectmore} on all level cuts $S_h,\ldots,S_1$. Moreover, for a node $v\in B_i$, the "`effective"' capacity of an arc from $v$ to $\bar{B}$ is $\min(i,1/\phi)$. Formally, we can bound $L_j$ by the total (effective) outgoing capacity, which is 
\vspace{-0.5em}
\begin{equation}
\label{eqn:outcap}
\sum_{i=1}^h |E(B_i,\bar{B})|\cdot \min(i,\frac{1}{\phi}) = \sum_{i=1}^{\frac{1}{\phi}} |E(S_i,\bar{B})|
\end{equation}
where $h$ is the bound on labels used in unit flow.

We design a charging scheme to charge the above quantity (the right hand side) to the mass in $\Delta_j(B)$, such that each unit of mass is charged $O(1/(\sigma_2\log \vol(B)))$. It follows that $L_j\leq O(\frac{1}{\sigma_2\log \vol(B)})\cdot |\Delta_j(B)|$. 

Recall that, $|E(S_i,\bar{B})|\leq \frac{|E(S_i,B\setminus S_i)|}{\sigma_2\log \vol(B)\log(1/\phi)}$ from Assumption~\ref{assum:connectmore}. We divide edges in $E(S_i,B\setminus S_i)$ into two groups: : $1)$ edges across one level, and $2)$ edges across more than one level. Let $z_1(i),z_2(i)$ be the number of edges in the two groups respectively.

If $z_1(i) \geq |E(S_i,B\setminus S_i)|/3$, we charge $3/(\sigma_2\log \vol(B)\log(1/\phi))$ to each edge in group $1$. These edges in turn transfer the charge to the absorbed mass at their endpoints in $B_i$. Since each node $v$ in level $i\geq 1$ has $d(v)$ absorbed mass, each unit of absorbed mass is charged $O(1/(\sigma_2\log \vol(B)\log(1/\phi)))$. Note that the group $1$ edges of different level $i$'s are disjoint, so each unit of absorbed mass will only be charged once this way.

If $z_1(i) \leq |E(S_i,B\setminus S_i)|/3$, we know $z_2(i)-z_1(i) \geq |E(S_i,B\setminus S_i)|/3$. Group $2$ edges in total send at least $(i-1)z_2(i)$ mass from $S_i$ to $B\setminus S_i$, and at most $(i-1)z_1(i)$ of these mass are pushed into $S_i$ by group $1$ edges. Thus, there are at least $(i-1)|E(S_i,B\setminus S_i)|/3$ mass that start in $S_i$, and are absorbed by nodes at level below $i$ (possibly outside $B$). In particular, this suggests $|M_j(S_i)|\geq (i-1)|E(S_i,B\setminus S_i)|/3$, and we split the total charge $|E(S_i,\bar{B})|$ evenly on these mass, so each unit of mass is charged $O(1/(i\sigma_2\log \vol(B)\log(1/\phi)))$. Since we sum from $i=1/\phi$ to $1$ in (RHS of) Eqn~\eqref{eqn:outcap}, we charge some mass multiple times (as $S_i$'s not disjoint), but we can bound the total charge by $\sum_{i=1}^{1/\phi} \frac{1}{i}\cdot O(1/(\sigma_2\log \vol(B)\log(1/\phi)))$, which is $O(1/(\sigma_2\log \vol(B)))$. This completes the proof.
\end{proof}

Now we fill in some details for the proof of Theorem~\ref{thm:recover}.
\recover*
\begin{proof}
Since we use $\phi=\Theta(\phi_S(B))$ when we call {\em CRD-inner}, we are guaranteed that {\em CRD-inner} will make enough effort to get through bottleneck of conductance $\Omega(\phi_S(B))$, so the diffusion should be able to spread mass completely over $B$, since any cut inside $B$ has conductance at least $\phi_S(B)$. Thus, there should be no excess remaining on nodes in $B$ at the end of {\em CRD-inner}, unless every node $v$ in $B$ has $m(v)$ mass already (i.e., all nodes in $B$ are {\em saturated}). 

Formally, consider the proof of Theorem~\ref{thm:CRD-inner}, but with everything with respect to the induced graph of $B$. If there is no excess pushed into $B$ from outside, we can use the exact arguments in the proof of Theorem~\ref{thm:CRD-inner} to show that either there is no excess on any node in $B$ at the end of the {\em CRD-inner} call, or there is a cut of conductance $O(\phi)$, or all nodes in $B$ are saturated. By assumption, the second case is not possible, so we won't remove excess supply between {\em CRD-inner} calls before all nodes in $B$ are saturated. If we consider supply pushed back into $B$ from outside, we can show that the amount of excess on nodes in $B$ at the end is at most the amount of mass pushed into $B$. Since we already counted the mass leaving $B$ as lost, we don't need to worry about them again when we remove the mass.

Consequently, before all nodes in $B$ are saturated, the amount of mass in $B$ only decreases (compared to the supposed $2d(v_s)\cdot 2^j$ amount in iteration $j$) due to mass leaving $B$, which we can bound ( Lemma~\ref{lemma:leakage}) by a $O(1/(\sigma\log \vol(B)))$ fraction of the mass in $B$ each iteration. We will have enough mass to spread over all nodes in $B$ in $O(\log \vol(B))$ iterations, so we lose $O(1/\sigma)$ fraction of mass before all nodes in $B$ are saturated. 

Once the diffusion has saturated all nodes in $B$, the amount of mass in $B$ will be $2\vol(B)$ at the start of every subsequent {\em CRD-inner} call. At most $\vol(B)\phi(B)/\phi_S(B)\leq O(\vol(B)/\sigma)$ mass can leave $B$, and nodes in $B$ can hold $\vol(B)$ mass, so there must be a lot of excess (in $B$) at the end. Thus, the CRD algorithm will terminate in at most $2$ more iterations after all nodes in $B$ are saturated, since the amount of mass almost stops growing due to excess removal.

At termination, the amount of mass is $\Theta(\vol(B))$, and only $O(1/\sigma)$ fraction of the mass is in $\bar{B}$. Since $S=\{v|m(v)\geq d(v)\}$, and the total mass outside is $O(\vol(B)/\sigma)$, we get claim $(1)$ of the theorem. In our simplified argument, all nodes in $B$ have saturated sinks (i.e., $\vol(B\setminus S)=0$) at termination. We get the small loss in claim $(2)$ when we carry out the argument more rigorously.

The amount of mass grows geometrically before the CRD algorithm terminates, so the running time is dominated by the last {\em CRD-inner} call. The total amount of mass is $O(\vol(B))$ in the last iteration, and the running time follows Theorem~\ref{thm:CRD-inner} with $\phi=\Theta(\phi_S(B))$
\end{proof}

%% file: appendix_empirical.tex
%
%
\section{Empirical Set-up and Results}
\label{sxn:empirical-appendix}
\subsection{Datasets}\label{subsec:datasets}
We chose the graphs of John Hopkins, Rice, Simmons and Colgate universities/colleges. 
The actual IDs of the graphs in Facebook100 dataset are Johns\_Hopkins55, Rice31, Simmons81 and Colgate88.
These graphs are anonymized Facebook graphs on a particular day in September 2005 for student social networks. 
The graphs are unweighted and they represent ``friendship ties''. The data form a subset of the Facebook100 dataset from \cite{TMP2012}. 
We chose these $4$ graphs out of $100$ due to their large assortativity value in the first column of Table A.2 in \cite{TMP2012},
where the data were first introduced and analyzed. Details about the graphs are shown is Table \ref{tab:data_collection}.
\begin{table}[htbp]
  \centering
{\scriptsize    \begin{tabular}{cccc}
    \toprule
    Graph & volume & nodes & edges \\
    \midrule
     John Hopkins& 373144 &5157 &186572 \\
     Rice& 369652 & 4083&184826 \\
     Simmons& 65968 & 1510&32984\\
     Colgate & 310086 & 3482&155043\\ 
    \bottomrule
    \end{tabular}%
    \caption{Graphs used for experiments.}
  \label{tab:data_collection}%
  }
\end{table}%

Each graph in the Facebook dataset comes along with $6$ features, i.e., second major, high school, gender, dorm, major index and year. 
We construct ``ground truth'' clusters by using the features for each node.
In particular, we consider nodes with the same value of a feature to be a cluster, e.g., students of year $2009$. We loop over all possible clusters
and consider as ground truth the ones that have volume larger than $1000$, conductance smaller than $0.5$ and gap larger than $0.5$.
Filtering results in moderate scale clusters for which the internal volume is at least twice as much as the volume of the edges that leave the cluster. 
Additionally, gap at least $0.5$ means that the smallest nonzero eigenvalue of the normalized Laplacian of the subgraph defined by the cluster is at least twice larger than the conductance of the cluster 
in the whole graph. The clusters per graph that satisfy the latter constraints are shown in Table \ref{tab:clusters}. 

We resort to social networks as our motivation is to test our algorithm against "noisy" clusters, and for social networks it is well known that reliable ground truth is only weakly related to good conductance clusters, and thus certain commonly-used notions of ground truth would not provide falsifiable insight into the method~\cite{Jeub15}. As analyzed in the original paper that introduced these datasets~\cite{TMP2012}, only year and dorm features give non-trivial "assortativity coefficients", which is a "local measure of homophily". This agrees with the ground truth clusters we find, which also correspond to features of year and dorm. 

\subsection{Performance criteria and parameter tuning}\label{subsec:param_tune}
For real-world Facebook graphs since we calculate the ground truth clusters in Table \ref{tab:clusters} then we measure performance by calculating precision and recall
for the output clusters of the algorithms.

We set the parameters of CRD to $\phi=1/3$ for all experiments.
At each iteration we use sweep cut on the labels returned by the {\em CRD-inner} subroutine to find a cut of small conductance, 
and over all iterations of CRD we return the cluster with the lowest conductance. 

ACL has two parameters, the teleportation parameter $\alpha$ and a tolerance parameter $\epsilon$. Ideally the former should be set according to the reciprocal of the mixing time 
of a a random walk within the target cluster, which is equal to the smallest nonzero eigenvalue of the normalized Laplacian for the subgraph that corresponds to the target cluster.
Let us denote the eigenvalue with $\lambda$. 
In our case the target cluster is a ground truth cluster from Table \ref{tab:clusters}. We use this information to set parameter $\alpha$. In particular, for each node in the clusters in Table \ref{tab:clusters}
we run ACL $4$ times where $\alpha$ is set based on a range of values in $[\lambda/2,2\lambda]$ with a step of $(2\lambda - \lambda/2)/4$. The tolerance parameter $\epsilon$ is set to $10^{-7}$ for all
experiments in order to guarantee accurate solutions for the PageRank linear system. For each parameter setting we use sweep cut to find a cluster of low conductance, and over all parameter 
settings we return the cluster with the lowest conductance value as an output of ACL.

For real-world experiments we show results for ACLopt. In this version of ACL, for each parameter setting of $\alpha$ we use sweep cut algorithm to obtain a low conductance cluster and then we compute its precision and recall.
Over all parameter settings we keep the cluster with the best F1-score; a combination of precision and recall. This is an extra level of supervision for the selection of the teleportation parameter $\alpha$, which is not possible in practice since it requires ground truth information.
However, the performance of ACLopt demonstrates the performance of ACL in case that we could make optimal selection of parameter $\alpha$ among the given range of parameters (which also includes ground truth information) for the
precision and recall criteria.

Finally, we set the reference set of FlowI to be the output set of best conductance of ACL out of its $4$ runs for each node.
By this we aim to obtain an improved cluster to ACL in terms of conductance. Note that FlowI is a global algorithm, which means 
that it accesses the information from the whole graph compared to CRD and ACL which are local algorithms.


\subsection{Real-world experiments}
For clusters in Table \ref{tab:clusters} we sample uniformly at random 
half of their nodes. For each node we run CRD, ACL and ACL+FlowI. We report the results using box plots, which graphically summarizes groups of numerical data using quartiles.
In these plot the orange line is the median, the blue box below the median is the first quartile, the blue box above the median is the third quartile, the extended long lines below and above the box are the maximum and minimum values and the circles are outliers.

The results for John Hopkins university are shown in Figure \ref{fig:john_hopkins}. 
Notice in this figure that CRD performs better than ACL and ACLopt, which both use ground truth 
information, see parameter tuning in Subsection \ref{subsec:param_tune}. CRD performs similarly to ACL+FlowI, where FlowI is a global algorithm, but CRD is a local algorithm.
Overall all methods have large medians for this graph because the clusters with dorm $217$ and year $2009$ are clusters with low conductance compared to the ones in other universities/colleges which we will discuss 
in the remaining experiments of this subsection.

The results for Rice university are shown in Figure \ref{fig:rice}. Notice that both clusters of dorm $203$ and year $2009$ for Rice university are worse in terms of conductance compared to the clusters 
of John Hopkins university. Therefore the performance of the methods is decreased. For the cluster of dorm $203$ with conductance $0.46$ CRD has larger median than ACL, ACLopt and ACL+Flow in terms of precision.
The latter methods obtain larger median for recall, but this is because ACL leaks lots of probability mass outside of the ground truth cluster since as indicated by its large conductance value many nodes in this cluster are connected 
externally. For cluster of year $2009$ CRD outperforms ACL, which fails to recover the cluster because it leaks mass outside the cluster, FlowI corrects the problem and locates the correct cluster at the expense of touching the whole graph.
Notice that all methods have a significant amount of variance and outliers, which is also explained by the large conductance values of the clusters.

The results for Simmons college are shown in Figure \ref{fig:simmons}. Notice that Simmons college in Table \ref{tab:clusters} has two clusters, one with poor conductance $0.47$ for students 
of year $2007$ and one low conductance $0.1$ for students of year $2009$. The former with conductance $0.47$ means that the internal volume is nearly half the volume of the outgoing edges. 
This has a strong implication in the performance of CRD, ACL and ACLopt which get median precision about $0.5$. This happens because the methods push half of the flow (CRD) and half of the probability mass (ACL) outside the ground truth cluster, 
which results in median precision $0.5$. ACL achieves about $20\%$ more (median) recall than CRD but this is because ACL touched more nodes than CRD during execution of the algorithm. Notice that ACL+FlowI fails for the cluster of year $2007$, this is because
FlowI is a global algorithm, hence it finds a cluster that has low conductance but it is not the ground truth cluster. The second cluster of year $2009$ has low conductance hence all methods have large median performance with CRD being slightly better than ACL, ACLopt and ACL+FlowI.

The results for Colgate university are shown in Figure \ref{fig:colgate}. The interesting property of the clusters in Table \ref{tab:clusters} for Colgate university is that their conductance varies from low $0.1$ to large $0.48$. 
Therefore in Figure \ref{fig:colgate} we see a smooth transition of performance for all methods from poor to good performance. In particular, for the cluster of year $2006$ the conductance is $0.48$ and CRD, ACL and ACLopt
perform poorly by having median precision about $50\%$, recall is slightly better for ACL but this is because we allow it touch a bigger part of the graph. ACL+FlowI fails to locate the cluster. For the cluster of year $2007$ the conductance 
is $0.41$ and the performance of CRD, ACL and ACLopt is increased with CRD having larger (median) precision and ACL having larger (median) recall as in the previous cluster. Conductance is smaller for the cluster of year $2008$, for which 
we observe substantially improved performance for CRD with large median precision and recall. On the contrary, ACL, ACLopt and ACL+FlowI have nearly $30\%$ less median precision in the best case and similar median recall, but only because
a large amount of probability mass is leaked and a big part of the graph is touched which includes the ground truth cluster. Finally, the cluster of year $2009$ has low conductance $0.11$ and all methods have good performance 
for precision and recall.
 
\begin{figure}
\centering
  \includegraphics[scale=0.59]{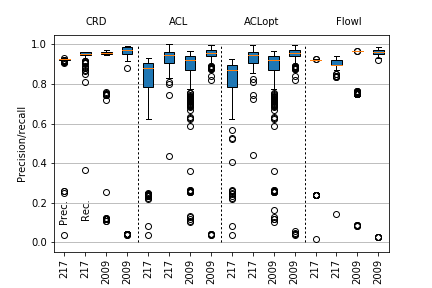}
\caption{Precision and recall results for John Hopkins university}
\label{fig:john_hopkins}
\end{figure}
\begin{figure}
\centering
  \includegraphics[scale=0.59]{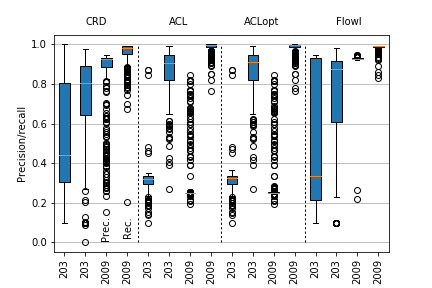}
\caption{Precision and recall results for Rice university}
\label{fig:rice}
\end{figure}
\begin{figure}
\centering
  \includegraphics[scale=0.59]{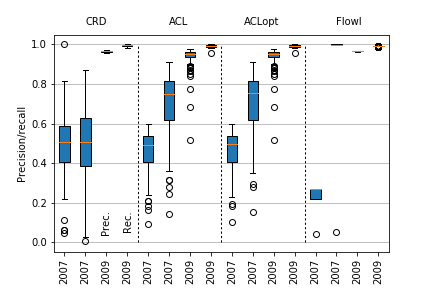}
\caption{Precision and recall results for Simmons college}
\label{fig:simmons}
\end{figure}
\begin{figure}
\centering
  \includegraphics[scale=0.59]{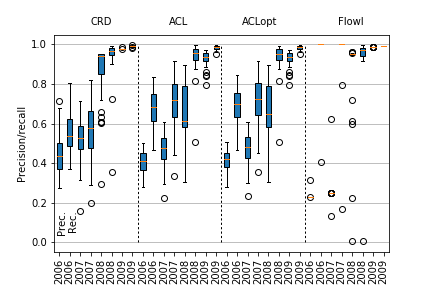}
\caption{Precision and recall results for Colgate university}
\label{fig:colgate}
\end{figure}